\RequirePackage{etex}
\documentclass[a4paper,12pt,reqno,oneside]{amsart}
\usepackage[multiple]{footmisc}

\usepackage{setspace}
\usepackage{amsmath}
\usepackage{amsfonts}
\usepackage{amssymb}
\usepackage{pictex}
\usepackage{}
\usepackage{amsthm}
\usepackage{breakcites}
\usepackage{graphics,epsfig,verbatim,bm,latexsym,url,amsbsy}
\usepackage{rotating}
\usepackage[authoryear,round]{natbib}
\usepackage{mathrsfs}
\usepackage{textgreek}
\usepackage{multirow}
\usepackage{graphicx}
\usepackage{subcaption}
\usepackage{array}
\usepackage{url}
\usepackage{graphicx}

\usepackage{float}
\usepackage{tikz}
\usepackage{subcaption}

\usepackage{xcolor}
\usepackage{bm}
\usepackage{tikz-cd}

\newcommand{\tr}{\mathsf{T}}

\usepackage{pstricks, enumerate, pst-node, pst-text, pst-plot}
\usepackage{xparse}

\newcommand{\dd}[2]{{\left\langle {#1},{#2}\right\rangle}}

\usepackage{graphicx}
\graphicspath{{screenshots/}{images/}} 

\theoremstyle{definition}

\long\def\symbolfootnote[#1]#2{\begingroup%
\def\thefootnote{\fnsymbol{footnote}}\footnote[#1]{#2}\endgroup}

\usepackage{footnote}
\makesavenoteenv{tabular}
\usepackage{fancyhdr}
\newcommand{\documenttitle}{Thesis}

\renewcommand{\max}{\operatornamewithlimits{max}}

\newcommand{\be}{\begin{equation}}
\newcommand{\ee}{\end{equation}}
\newcommand{\bes}{\begin{equation*}}
\newcommand{\ees}{\end{equation*}}

\usepackage{amsmath}
\usepackage{epsfig,graphics}
\usepackage{etoolbox}
\usepackage[breaklinks]{hyperref}
\usepackage{accents}

\newcommand{\change}[1]{\widehat{#1}}

\newcommand{\rdots}{\mathinner{%
  \mkern1mu\raise1pt\hbox{.}%
  \mkern2mu\raise4pt\hbox{.}%
  \mkern2mu\raise7pt\vbox{\kern7pt\hbox{.}}\mkern1mu}}

\DeclareMathOperator{\Var}{Var}

\usepackage{etex}
\DeclareFontFamily{U}{mathx}{\hyphenchar\font45}
\DeclareFontShape{U}{mathx}{m}{n}{
      <5> <6> <7> <8> <9> <10>
      <10.95> <12> <14.4> <17.28> <20.74> <24.88>
      mathx10
      }{}
\DeclareSymbolFont{mathx}{U}{mathx}{m}{n}
\DeclareFontSubstitution{U}{mathx}{m}{n}
\DeclareMathAccent{\widecheck}{0}{mathx}{"71}
\DeclareMathAccent{\wideparen}{0}{mathx}{"75}

\usepackage{palatino}
\usepackage[margin=1in]{geometry}

\hypersetup{colorlinks=true, urlcolor=black!50!blue, citecolor=black!50!blue, linkcolor=black!50!blue}

\makeatletter
\def\@footnotecolor{purple!30!blue}
\define@key{Hyp}{footnotecolor}{%
 \HyColor@HyperrefColor{#1}\@footnotecolor%
}
\patchcmd{\@footnotemark}{\hyper@linkstart{link}}{\hyper@linkstart{footnote}}{}{}
\makeatother
\hypersetup{footnotecolor=black}

\usepackage{amsthm, thmtools}
\usepackage{
nameref,
hyperref,
cleveref,
}

\newcommand\nc{\newcommand}
\nc\on{\operatorname}
\theoremstyle{definition} \newtheorem{thm}{Theorem}
\theoremstyle{definition} \theoremstyle{definition}
\theoremstyle{definition} 
\theoremstyle{definition} \theoremstyle{remark}
\theoremstyle{definition} 
\theoremstyle{definition} 
\theoremstyle{definition} \theoremstyle{plain}
\theoremstyle{definition} 
\theoremstyle{definition} \newtheorem{lem}{Lemma}
\theoremstyle{definition} 
\theoremstyle{definition} \newtheorem{prop}{Proposition}
\theoremstyle{definition} 
\theoremstyle{definition} \newtheorem{assumption}{Assumption}
\theoremstyle{definition} \newtheorem{example}{Example}
\theoremstyle{definition} 
\newtheorem*{PropertyA}{Property A}
\theoremstyle{definition} 
\theoremstyle{definition} 
\theoremstyle{definition}

\theoremstyle{definition}

\allowdisplaybreaks

\title{Taxes and Market Power: A Principal Components Approach}

\date{\today}
\author[]{Andrea Galeotti \and Benjamin Golub \and Sanjeev Goyal  \and \\  Eduard Talam\`{a}s \and Omer Tamuz}\thanks{ Galeotti is at the London Business School, Golub is at Northwestern University, Goyal is at University of Cambridge, Talam\`{a}s is at IESE Business School, and Tamuz is at Caltech. This work was supported by the ERC grant 724356 (Galeotti), the Pershing Square Fund for Research on the Foundations of Human Behavior (Golub), the National Science Foundation (SES-1658940, SES-1629446, Golub \& DMS-1944153, Tamuz), the Sloan Foundation (Tamuz) and the US-Israel Binational Science Foundation (2018397, Tamuz). }

\begin{document}

\maketitle

\begin{abstract}
Suppliers of differentiated goods make simultaneous pricing decisions, which are strategically linked. Because of market power, the equilibrium is inefficient. We study how a policymaker should target a budget-balanced tax-and-subsidy policy to increase welfare. A key tool is a certain basis for the goods space, determined by the network of  interactions among suppliers. It consists of \emph{eigenbundles}---orthogonal in the sense that a tax on any eigenbundle passes through \emph{only} to its own price---with pass-through coefficients determined by associated eigenvalues. Our basis permits a simple characterization of optimal interventions. A planner maximizing consumer surplus should tax eigenbundles with low pass-through and subsidize ones with high pass-through. The \emph{Pigouvian leverage} of the system---the gain in consumer surplus achievable by an optimal tax scheme---depends only on the dispersion of the eigenvalues of the matrix of strategic interactions. We interpret these results in terms of the network structure of the market.
\end{abstract}

\section{Introduction}

Imperfect competition has attracted renewed attention in recent years, as evidence accumulates of the rising importance of oligopolistic industries  \citep{de2020rise,azarvives2021}. 
As firms in such industries have market power, the market outcome is typically inefficient.\footnote{For recent empirical evidence on the size of the welfare loss see \citet{pelligrino2021} and \citet{edererpelligrino2021}.} This creates the scope for targeted interventions that tax some suppliers/sectors and subsidize others to further a social objective. Surprisingly, very little is known about this classical Pigouvian problem  for oligopolies. This paper addresses it through the lens of networks.

In our model, firms price  their products strategically (in Bertrand competition), and demand arises from the behavior of a representative consumer. Some products may be substitutes and others may be complements. The structure of such relationships determines the strategic interactions among firms. This, in turn, determines the Nash equilibrium prices, how surplus is allocated among producers, and consumer surplus. 

The implications of taxing or subsidizing the suppliers in these markets are complicated, since changes to firms' costs affect the prices and quantities of other firms' through the strategic competition firms face. We summarize the strategic interactions among suppliers in a \emph{spillover matrix} $\bm{D}$ where $D_{ij}$ specifies how the demand of firm $i$ changes when the price of firm $j$ changes. 
We think of $\bm{D}$ as describing a \emph{network structure} on the firms, following a literature in network theory (e.g., \citet{Candoganetal2012} and \citet{BlochQuerou2013}), and analyze certain properties of this network to facilitate the study of optimal tax policies.

The spillover matrix induces a basis of eigenvectors (also called principal components) of the space of goods produced by the strategic firms. We call the vectors in this basis \emph{eigenbundles}. This basis has three special properties. First, when the cost of one eigenbundle is changed (e.g., by taxation), the effect is to change equilibrium prices \emph{only of that eigenbundle}. These eigenbundles therefore identify independent, or orthogonal, dimensions of the market, such that the costs of one eigenbundle do not affect the prices of others. Second, the price pass-throughs (i.e., the fraction of a cost increase that is passed to the consumer) associated with various eigenbundles can be calculated in terms of corresponding eigenvalues of the spillover matrix. Third, the market induces a ranking of pass-throughs: the eigenbundles with larger eigenvalues---which are more representative (in a precise sense) of the demand interdependence across products offered in the market---are those with smaller pass-throughs.

These properties facilitate simple expressions of the effect of tax-subsidy schemes on prices and welfare. That, in turn, permits intuitive  characterization of interventions that maximize consumer surplus. We focus on optimal \emph{small} changes to taxes. Theoretically, the analysis of optimal small interventions requires fewer assumptions  about the structure of demand and technology at allocations far from current conditions than a global analysis would, and comes down to an analysis of marginal effects. Practically, planners may prefer  small adjustments to the status quo rather than implementing a large change all at once---for example, due to political constraints, or because of risk aversion. 

Our study of the small-intervention problem builds on a large literature on the so-called ''tax reform approach'' initiated by \citet{feldstein1976theory} and  \citet{diewert1978optimal}, \citet{dixit1979price} and \citet{tirole1981tax}. Our first contribution here is to formalize the incentives making the planner choose small interventions. We posit that reform is implemented with some noise and the policymaker is risk-averse over the welfare criterion of interest (e.g., consumer surplus). As the risk-aversion increases, the planner prefers smaller interventions. We formalize this problem and characterize optimal tax interventions.

We obtain the following economic insights. First, through the lens of our decomposition, the optimal intervention may be described as follows: High pass-through eigenbundles are subsidized, with revenue collected from low pass-through eigenbundles. Thus, taxes levied on the low pass-through eigenbundles are mostly paid by the firms, while the subsidies offered to the high pass-through eigenbundles are mostly passed on to the consumer. This results in increased consumer welfare. As we have mentioned, the ranking of eigenbundles by pass-throughs is the same order of the corresponding eigenvalues of the spillover matrix.

Second, we characterize the gain in surplus that a planner can achieve through a small tax intervention. This depends on the scale of the intervention that a planner is willing to implement, but beyond this it also depends on the network: some networks permit the same planner to achieve a much greater gain, offering a bigger ''bang for the buck.''  The name we give to this ''bang for the buck'' measure is the \emph{Pigouvian leverage}.

Our third contribution is to show that the \textbf{Pigouvian leverage} depends in a simple way on the eigenvalues of the spillover matrix: It is proportional to the sample variance of the eigenvalues of this matrix. An equivalent characterization is that the \textbf{Pigouvian leverage} is proportional to the sum of the squares of the off-diagonal elements of the spillover matrix, which is a natural measure of the intensity of strategic effects. The ''overall magnitude'' of spillovers thus determines how much scope there is for improving consumer surplus. Notably, the \emph{signs} of strategic interactions do not matter for this measure. Whether firms' prices are strategic complements or substitutes, it is only an aggregate measure of the intensity of strategic interactions that determines the Pigouvian leverage. This gives an economically substantive characterization of when planners can achieve large benefits for consumers, but also illustrates the value of the spectral approach in generating such economic insights.

We conclude by exploring \textbf{non-small optimal interventions}. To do this, we focus on the classical case of \textbf{linear demand}. The characterization is again in terms of the spectral decomposition of demand spillovers, and echoes the result in the small intervention regime. The optimal policy still collects tax revenue from the eigenbundles with low pass-through, where the impact on prices and output is relatively small, and allocates them---via subsidies---to the eigenbundles with high pass-through, where the impact on prices and output is relatively large. The fact that the deadweight loss from taxation increases as a quadratic function of its size determines the optimal size of this tax program. 

\subsection{Related literature}

Our paper contributes to a literature on the structure and theoretical properties of market power. For an early theoretical  paper see \citet{dixit1986comparative}; more recent studies include, for example, \citet{vives1999oligopoly} and \citet{azarvives2021}.  A recent literature in macroeconomics and industrial organization uses network models of differentiated oligopoly, with models similar to the one studied here, to provide empirical estimates of  welfare losses due to market power (see e.g., \citet{pelligrino2021} and \citet{edererpelligrino2021}).\footnote{See also \citet{elliott2019role} for related arguments on how network methods can be useful for competition authorities in developing antitrust investigations.}

Our paper makes two contributions to this literature. First, taking a network perspective, we provide a geometric approach to analyzing pass-through in the pricing game with market power.   This builds on work  emphasizing the value of pass-through as a conceptual tool (e.g., \citet{weyl2013pass} and \citet{miklos2021pass}), and shows how it can be described tractably and intuitively in markets with very rich heterogeneity.  Second, motivated by the empirical research on welfare losses due to market power, we apply our decomposition to characterize  policy interventions that maximize consumer surplus and shed light on the economic forces that make them most effective. 

In taking a network approach to inefficiencies in market interactions among firms we also connect to work in macroeconomics and trade on production networks (e.g., \citet{baqaee2018}, \citet{grassi2017} and \citet{liu2019} and \citet{grassisauvagnat2019}). This work highlights the importance of the interaction between market structure and production networks for a variety of outcomes, especially the amplification of productivity shocks. One important contrast is that those models take pricing distortions as given, rather than endogenizing prices as we do in a Bertrand model.



More generally, our paper contributes to the theory of network interventions; prominent early contribution to this theory include \citet{Borgatti2006} and \citet*{Ballester2006}.\footnote{The literature on this subject is very large. Other contributions of network intervention in models of information diffusion, advertising, and pricing include \citet*{banerjee2013diffusion}, \citet{BelhajDeroian}, \citet{BlochQuerou2013}, \citet*{Candoganetal2012}, \citet{GDemange}, \citet{FainmesserGaleotti2017}, \citet{GaleottiGoyal2009}, \citet{GaleottiRogers}, and \citet*{leduc2017pricing}.} In a recent paper, \citet*{galeotti2020targeting} study intervention in a certain canonical class of quadratic network games. They use the singular value decomposition of the interaction matrix for the study of (costly) interventions that alter the stand-alone marginal benefits of individual activity. Their intervention problem takes a stylized form that cannot (even approximately) accommodate the notion of a budget-balanced tax scheme. As a result, the characterizations of optimal policy in our setting are distinctive, for example the finding that the planner subsidizes some eigenbundles and taxes others. At a conceptual level, we introduce the idea of Pigouvian leverage---how much the planner can achieve with budget-balanced interventions. This notion, and the characterizations that we provide of how it relates to the market structure, are one of the main contributions of the present paper.

\section{Price pass-through and oligopoly competition}

Symbols denoting vectors and matrices are in bold. For any matrix $\bm{M}$, the symbol $m_{ij}$ stands for its element in the $i$th row and $j$th column, and $\bm{M}^\tr$ denotes its transpose.
The symbol $\dd{\bm{a}}{\bm{b}}$ denotes the dot product of $\bm{a}$ and $\bm{b}$.


\subsection{Market structure} \label{sec:structure}

The set of \emph{goods}  is $\mathcal{N}=\{1,2,\ldots,N\}$. Good $i$ is produced  by firm $i$. The firms engage in a Bertrand pricing game, simultaneously choosing prices $\bm{p}=(p_i)_{i \in \mathcal{N}}$. Firm $i$ has profit function 
$$ 
  \pi_i(\bm{p}) = q_i(p_i - c_i),
$$
where $q_i$ is the quantity demanded of good $i$, and $c_i$ is the (constant) marginal cost of production.

There is a representative consumer with a quasi-linear utility function that depends on goods, described by a quantity vector $\bm{q}$, and a separate commodity called the numeraire, whose quantity is denoted by $m$:
$$
  \mathcal{U}(\bm{q},m) = \mathcal{V}(\bm{q}) + m.
$$ Here $\mathcal{V}$ is a twice-differentiable, concave utility function.

\begin{example}\label{ExQU} For a concrete example, the reader may keep in mind the classical case in which the utility of the representative consumer is quadratic in consumption:
$$\mathcal{V}(\bm{q}) = \bm{\beta}^\tr \bm{q}-\frac{1}{2}\bm{q}^\tr \bm{B} \bm{q},$$ where $\bm{\beta}$ is a vector with positive entries and $\bm{B}$ is a symmetric, positive definite matrix. \end{example}

Returning to the general case, the consumer chooses an optimal bundle, paying price $p_i$ for commodity $i$ in terms of the numeraire. We assume that $\mathcal{V}(\cdot)$ is such that, for any  price vector $\bm{p}$, there is a unique, interior demand profile solving the consumer's problem
\begin{equation}\label{eq:utilitymax}
\max_{\bm{q}} \; \; \mathcal{V}(\bm{q})-\langle\bm{q},\bm{p}\rangle.
\end{equation}
Let $q_i(\bm{p})$ be the quantity demanded of good $i$ at the profile of prices $\bm{p}$. 


\subsection{Equilibrium}

We focus on an (interior pure-strategy Nash) equilibrium $\bm{p}^*$ of the price-setting game between the firms. 
The first-order conditions that characterize an interior equilibrium are:
\begin{equation}\label{FOC}
q_i(\bm{p}^*)+\frac{\partial q_i(\bm{p}^*)}{\partial p_i}(p_i^*-c_i)=0 \quad \text{for all}\  i\in\mathcal{N}.
\end{equation}

We take the functions $q_i$ as fixed and take the parameters of the game to be the marginal costs $\bm{c}$ (which will be modified by tax/subsidy interventions). We make the following assumptions for technical convenience.
\begin{assumption} \label{ass:eqmexist} There is an open set $\mathfrak{C}$ of costs such that, for all $\bm{c} \in \mathfrak{C}$,
\begin{enumerate}
\item an interior equilibrium $\bm{p}^*$  exists and is unique;
\item the demand function $\bm{q}(\bm{p})$ is locally linear in a neighborhood of equilibrium prices $\bm{p}^*$;
\item each $q_i(\bm{p})$ is strictly decreasing in each $p_i$. 
\end{enumerate} \end{assumption}

We assume throughout the analysis that the status quo is at some $\bm{c} \in \mathfrak{C}$, and we then perturb the market relative to that starting point. This model is a standard differentiated oligopoly (e.g., \citet{vives1999oligopoly} and \citet{chone2020linear}). Firms produce products that can be complements and substitutes in consumption and, in turn, this determines strategic interactions across the firms. 

\subsection{Pass-through}\label{sec:Basic_pass}

Slightly abusing terminology, we often denote equilibrium prices $\bm{p}^*$ and quantities $\bm{q}(\bm{p}^*)$ by $\bm{p}$ and $ \bm{q}$, respectively. 

We are interested in how a change in production costs (e.g., because of taxation) passes through to equilibrium prices, and how it changes welfare. To this end,  we introduce a dummy parameter $\xi$ and posit that costs vary according to given functions $\bm{c}(\xi)$. We let $\dot{x}$ denote $\frac{dx}{d\xi}$ for any variable $x$ in the model.  The vector $$ \dot{\bm{c}}=(\dot{c}_1, \dot{c}_2, \dots, \dot{c}_N),$$ which will be specified exogenously, determines the first-order response of the economy. To analyze this response, we will calculate the derivatives in $\xi$ of other variables in the model.

Totally differentiating (\ref{FOC}) around the equilibrium $\bm{p}$ yields, {using the local linearity assumption}:
\begin{equation}\label{totaldifferentiation1}
\sum_{j\in \mathcal{N}} \frac{\partial q_i(\bm{p})}{\partial p_j} {\dot{p}_j}  +\frac{\partial q_i(\bm{p})}{\partial p_i}(\dot{p}_i-\dot{c}_i)=0.
\end{equation}
To write this equation in vector notation, we  define $$ D_{ij}(\bm{p}) = \frac{\partial q_i(\bm{p})}{\partial p_j} \quad \forall i,j\in \mathcal{N},$$ which is element $ij$ of the the Jacobian of the vector $\bm{q}$ in the prices $\bm{p}$; let $\bm{D}$ be the respective $N\times N$ matrix.

Because demand arises from the preferences of a representative consumer (with twice-differentiable, concave utility), the matrix $\bm{D}$ is symmetric and negative semidefinite \citep{nocke2017quasi}. We also introduce a normalization that is useful in the rest of our analysis; Appendix \ref{sec:normalization} shows that this normalization is without loss of generality, as it can be achieved by changing the units of quantity of each good. We summarize this in the following property:
\begin{PropertyA}\label{Ass:psd} $\bm{D}$ is a symmetric negative semidefinite matrix and it satisfies the normalization $D_{ii}=-1$ for each $i\in \mathcal{N}$.
\end{PropertyA}

\autoref{totaldifferentiation1} becomes
\begin{equation}\label{totaldifferentiation2}
\dot{p}_i =\dot{c}_i + \sum_{j\in \mathcal{N}} D_{ij} {\dot{p}_j}, 
\end{equation}
which illustrates how the the strategic relations between any two suppliers $i$ and $j$ are captured by the sign of $D_{ij}$. We shall say that two distinct goods $i$ and $j$ are \textit{strategic substitutes} if $D_{ij}$ is negative and \textit{strategic complements} if $D_{ij}$ is positive. (Note that the strategic spillover is ``opposite'' to the interaction of the goods in the consumer's utility: if two goods are complements in consumption, they are strategic substitutes in the pricing game.)
The prices pass-through can be expressed in matrix form as follows:
\begin{equation}
\label{PT}
[\bm{I}-\bm{D}]\dot{\bm{p}}=\dot{\bm{c}}.
\end{equation}

Note that when demand is linear (the utility of the representative consumer is quadratic in consumption as in Example \ref{ExQU}), the pass-through equation \ref{PT} holds exactly when ${\dot{\bm{c}}}$ is understood as an arbitrary change.

\subsection{Pass-through in terms of eigenbundles}\label{sec:pass}

We now express the pass-throughs of cost changes in a compact way by changing to a convenient basis of \emph{eigenbundles}.

Property A implies that $\bm{D}$ is orthogonally diagonalizable. That is, there exists an $N\times N$ orthonormal matrix  $\bm{U}$ such that $$\bm{D}=\bm{U} \bm{\Sigma} \bm{U}^\tr,$$ where $\bm{\Sigma}$ is an $N\times N$ diagonal matrix whose $\ell^{\text{th}}$ diagonal element is the $\ell^{\text{th}}$-largest eigenvalue of $\bm{D}$, called $\sigma_\ell$; it is nonpositive because $\bm{D}$ is negative semidefinite. The ${\ell}^{\text{th}}$ column $\bm{u}^{\ell}$ is the eigenvector of $\bm{D}$ corresponding to $\sigma_\ell$.  We call this the ${\ell}^{\text{th}}$ \emph{eigenbundle} of $\bm{D}$. These vectors have norm $1$ and are orthogonal to each other.\footnote{Eigenbundles may have entries with both positive and negative signs and these signs can be interpreted as combining long holdings of goods (those with positive signs) and short holdings of goods (those with negative signs).}

For any $\bm{x}\in \mathbb{R}^N$, let $\underline{\bm{x}}$ denote $\bm{U}^\tr \bm{x}$; that is, $\underline{\bm{x}}$ is the vector $\bm{x}$ expressed in the basis $\bm{U}$.
We choose the signs of $\bm{u}^{\ell}$ so that the equilibrium quantities of eigenbundles are all positive, i.e., $\underline{q}_{\ell}\geq 0$ for each $\ell \in \mathcal{N}$.

The \emph{consumer surplus} associated to an price profile $\bm{p}$ is defined to be $$\mathcal{C}=\mathcal{V}(\bm{q}(\bm{p}))-\langle\bm{q}(\bm{p}),\bm{p}\rangle.$$ This is the net utility to the consumer after paying for the goods.

\begin{prop} \label{prop:eig} Fixing any $\dot{\bm{c}}$, the derivative of the equilibrium price of the $\ell^{\text{th}}$ eigenbundle has the following form:
\begin{equation} \label{eq:key}
\dot{\underline{{p}}}_{\ell}=\lambda_{\ell} \dot{\underline{{c}}}_{\ell}, 
\end{equation} 
where $$\lambda_\ell=\frac{1}{1-\sigma_\ell}$$ is increasing in $\ell$.
The associated pass-through to consumer surplus is
$$
  \dot{\mathcal{C}}=- \sum_\ell \lambda_{\ell}\underline{q}_{\ell} \dot{\underline{{c}}}_{\ell}. 
$$
\end{prop}

\begin{proof} 
From (\ref{PT}) we get
$\left(\bm{I}- \bm{U} \bm{\Sigma} \bm{U}^\tr \right) \dot{\bm{{p}}}=\dot{\bm{{c}}}.$
Multiplying both sides by $\bm{U}^\tr$ we get $$\bm{U}^\tr\left(\bm{I}- \bm{U} \bm{\Sigma} \bm{U}^\tr \right)\bm{U} \bm{U}^\tr \dot{\bm{{p}}}=\bm{U}^\tr \dot{\bm{{c}}},$$ that is, $\dot{\underline{\bm{{p}}}}=(\bm{I}-\bm{\Sigma})^{-1} \dot{\underline{\bm{{c}}}}.$ 

Turning to consumer surplus, using \autoref{eq:utilitymax} we have that
\begin{eqnarray*}
\dot{\mathcal{C}}&=&\sum_i \left(\frac{\partial \mathcal{V}(\bm{q})}{\partial q_i}\dot{q}_i-q_i\dot{p}_i-\dot{q}_ip_i\right)\\
&=&\sum_i \left(p_i\dot{q}_i-q_i\dot{p}_i-\dot{q}_ip_i\right)\\
&=&-\langle \dot{\bm{p}},\bm{q} \rangle,
\end{eqnarray*}
where the second equality follows by the envelope theorem applied to \autoref{eq:utilitymax} (which implies that $\frac{\partial \mathcal{V}(\bm{q})}{\partial q_i}=p_i$). Changing to the basis of $\bm{U}$ gives the result.\end{proof}

The usefulness of the eigenbundles is that---as \autoref{prop:eig} states---a change in the cost of one of them affects only its own prices with a certain coefficient $\lambda_\ell$, called the \emph{price pass-through}. The price pass-throughs of the eigenbundles are ordered according to their corresponding eigenvalues: the larger the $\ell^{\text{th}}$ eigenvalue $\sigma_\ell$, the higher the pass-through from changes in the cost of the $\ell^{\text{th}}$ eigenbundle to its equilibrium price.

The following example illustrates the proposition. 

\subsubsection{Example} \label{sec:ex}
Consider an oligopoly model with three products. Product $1$ and $2$ are independent; product $3$ is a complement to product $1$ and to product $2$. For example, product 1 is a video game, product 2 is a word processor, while product 3 is a computer. This represented by the following matrix (which is normalized):
	$$
	 \bm{D}= - \left[ \begin{array} {ccc} 1 & 0 &1/\sqrt{2}\\ 0 & 1&1/\sqrt{2}\\ 1/\sqrt{2} & 1/\sqrt{2}&1 \end{array}\right]
	$$
Figure \ref{fig:sim} depicts the network associated with $\bm{D}$ (omitting self-loops). Note that even if product $1$'s demand is independent of product $2$'s price and vice-versa, a shock that affects, say, the price of product $1$, will change product $3$'s price and, consequently, will affect the pricing strategy of the producer of product $2$.

\begin{figure}
		\begin{tikzpicture}[scale=.5]
\node[rectangle, draw] (s1) at (-4,4) {Product $1$ (e.g., video game)};
\node[rectangle, draw] (s2) at (12,4) {Product $2$ (e.g., word processor)};

\node[rectangle, draw] (s3) at (4,0) {Product $3$ (e.g., computer)};
\draw[very thick] (s1)--(s3);
\draw[very thick]  (s2)--(s3);
\end{tikzpicture}
\caption{A three-product oligopoly.
} \label{fig:sim}
\end{figure}
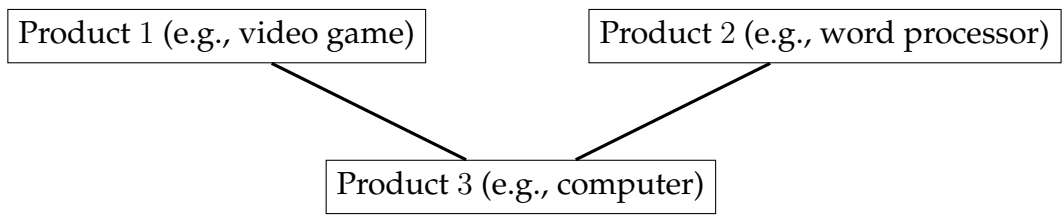

Proposition \ref{prop:eig} tells us that we can take any shock in the economy and decompose it into three shocks---each proportional to one the three eigenbundles of $\bm{D}$---and study the pass-through within each eigenbundle independently. We do this next. We will then compute the overall price pass-through of the shock to each of the three product as a combination of the price pass-throughs within each eigenbundle.

The first eigenbundle of $\bm{D}$ is  $\bm{u}^1=(1/2,1/2,-1/\sqrt{2})$. A cost shock proportional to $\bm{u}^1$ changes the prices of the two independent goods in the same direction and changes the price of good $3$ in the opposite direction. The strategic substitutability between these products amplifies this divergence in the price of the two independent products viz. product $3$. Initial shocks are amplified. In fact, since the eigenvalue $\sigma_1=0$, we have that $\lambda_1=1$ and so there is a $100\%$ price pass-through along this dimension.

The second eigenbundle of $\bm{D}$ is  $\bm{u}^2=(1/\sqrt{2},-1/\sqrt{2},0)$. A cost shock proportional to $\bm{u}^2$ moves prices of the two independent goods in opposite directions---without affecting the price of product $3$. Moreover, as product $3$ is complementary to both $1$ and $2$ (in a symmetric way) the opposite price changes of $1$ and $2$ keeps the price of producer $3$ constant. We obtain that pass-through in this eigenbundle is like the (normalized) monopoly pass-through of $1/2$. Indeed, the second eigenvalue is $\sigma_2=-1$, and so $\lambda_2=1/2$.

Finally, the third eigenbundle of $\bm{D}$ is  $\bm{u}^3=(1/2,1/2,1/\sqrt{2})$. It ranks products in terms of their eigenvector centrality. A cost shock proportional to $\bm{u}^3$ will therefore co-move the price of the three producers in the same direction. These effects are dampened by their strategic reaction, so the resulting pass-through to prices is relatively small. Indeed, the eigenvalue is $\sigma_3=-2$ and so only $\lambda_3=1/3$ of the shock passes through to the price of the third eigenbundle.

We now derive the price implications  of a specific marginal shock to costs---as well as the overall effect of this shock on consumer surplus. Consider a shock that marginally increases the cost of product $3$, i.e., $\dot{\bm{c}}=\{0,0,\dot{c}_3\}$. Then we have that $\dot{\bm{\underline{c}}}=\{-\dot{c}_3/\sqrt{2},0,\dot{c}_3/\sqrt{2}\}$ and $\dot{\bm{\underline{p}}}=\{-\dot{c}_3/\sqrt{2},0,\dot{c}_3/(3\sqrt{2})\}$.

The final price pass pass-throughs to the three products can be obtained by combining the different eigenbundles' pass-throughs with how represented each of the products is in them, i.e., $\dot{p}_{i}=\sum_\ell \bm{u}^\ell_i \dot{\bm{\underline{p}}_\ell}$. That is
\begin{eqnarray*}
\dot{p}_1=\dot{p}_2&=&\frac{\dot{c}_3}{2}\left(\frac{1}{3\sqrt{2}}-\frac{1}{\sqrt{2}}\right)=- \frac{\dot{c}_3}{3\sqrt{2}}\\
\dot{p}_3&=&\frac{\dot{c}_3}{\sqrt{2}}\left(\frac{1}{3\sqrt{2}}+\frac{1}{\sqrt{2}}\right)=\frac{2  \dot{c}_3}{3}
\end{eqnarray*}
Finally, the resulting change in consumer surplus is approximately the sum of the price change within each eigenbundle times the equilibrium demand of this eigenbundle. The first eigenbundle is shocked negatively, $\underline{\dot{\bm{c}}}_1=-\dot{c}_3/\sqrt{2}$, so the prices of the two independent goods decrease and the price of the complement good 3 increases, increasing  consumer surplus by $\dot{c}_3\underline{q}_1/\sqrt{2}$. The second eigenbundle is not shocked, i.e., $\underline{\dot{\bm{c}}}_2=0$, and so there is no effect on consumer surplus. The cost shock to the third eigenbundle $\underline{\dot{\bm{c}}}_3=\dot{c}_3/\sqrt{2}$ is positive and it increases prices and so it has a negative effect on consumer surplus ($-\dot{c}_3\underline{q}_3/(3\sqrt{2})$). Aggregating these three effects we obtain
\[
\dot{\mathcal{C}}=\frac{\dot{c}_3}{\sqrt{2}}\left(\underline{q}_1-\frac{1}{3}\underline{q}_3\right)=-\frac{\dot{c}_3}{3\sqrt{2}}\left(\frac{4}{\sqrt{2}}q_3-(q_1+q_2)\right)
\]
If the initial equilibrium demand for the complementary product $3$ is at least (roughly) 35\% of the total quantity demanded of product $1$ and product $2$, then the increase in the cost of product $3$ has a negative effect on consumer surplus; otherwise the effect is positive. In particular, when the demand for products 1 and 2 is relatively large, an increase in the cost of product 3 can---via the associated strategic response of the producers of products 1 and 2, who decrease prices---increase consumer welfare. 

\subsubsection{Comments on the meaning of eigenbundles} 
We now develop the intuition for the decomposition of pass-through into eigenbundles, building on the ideas shown in the example. For concreteness, we consider a case where all products are strategic substitutes ($D_{ij} \leq 0$ $\forall i,j$, corresponding to complements in consumption), and we focus on the first and the last eigenbundles.

The last eigenbundle $\bm{u}^n$ (the one associated with the highest eigenvalue in absolute value) has all entries positive. It corresponds to the eigenvector centralities in $-\bm{D}$ of various products. Producer $i$ is highly central---that is, ${u}^n_i$ is high---when the good that she produces is highly complementary to other goods which are also highly complementary to other goods, etc. 

Consider a shock that changes the producers' marginal costs proportionally to their eigenvector centralities. The first effect is to move the price of all producers in the same direction. Since producers sell complementary products, the underlying pricing game is one of strategic substitutes. Hence, the initial co-movement in price will be attenuated by producers' strategic responses. Since the first eigenbundle captures the extent to which each product is complementary to others taking into account the whole structure of feedback, this is the eigenbundle in which strategic responses  most strongly dampen the price pass-through of a shock.

In contrast, the first eigenbundle $\bm{u}^1$ reflects a very different aspect of the structure of $\bm{D}$. Its entries have \emph{opposite} signs for complementary goods insofar as possible.\footnote{See \citet{bramoulle2014strategic} and \citet{galeotti2020targeting} for extended discussions of the structure of such vectors.}
A shock that is proportional to $\bm{u}^1$ typically changes the prices of complementary goods in opposite directions. Since prices of such goods are strategic \emph{substitutes}, these initial price effects are amplified; second-order effects go in the same direction as first-order effects, reinforcing them further. So the corresponding pass-through to prices is large.

\section{Piguvian intervention: Small tax reforms}

Because firms have market power, the equilibrium will typically be inefficient. Thus there is scope for a planner to suitably choose taxes and subsidies at the firm level to improve welfare. We consider Pigouvian interventions that attempt to increase consumer surplus, but the analysis can be easily extended to other surpluses. We also focus on interventions that are:

\begin{itemize}

\item {\bf Budget neutral.} A planner with a loose budget constraint can subsidize all firms at a high rate and increase consumer surplus substantially. However, under a tight budget constraint, the planner will need to impose positive taxes on some firms in order to subsidise others. To understand these trade-offs we focus on the case that the planner must run a balanced budget; allowing the planner to run a small deficit will not change the qualitative nature of the results.

\item {\bf Small.} There are two reasons to focus on small interventions. One is realism, as an administrative authority may be unwilling to make large changes.  Secondly, if changes are constrained to be small, only properties of the market local to the equilibrium outcome are relevant, and these are much more feasible to assess. In fact, by focusing on small interventions we can derive economic insights which do not depend on the specific parametrization of the demand function of the representative consumer. As we shall see these results capture properties of global interventions that can be derived in special cases; see Secsion \ref{sec:linear} where we develop a non-small optimal intervention for the special case of linear demand. 
\end{itemize}

We begin by giving a simple foundation under which a planner would indeed choose to enact small changes and then we characterise the optimal intervention.

Let $\bm{\tau}=(\tau_1,\dots, \tau_n)$ be the profile of per-unit taxes that are assessed. This means that if ${q}_i$ units are produced of input $i$, then the planner collects $\tau_i q_i$ units of revenue, which may be negative if $\tau_i <0$. The planner does not have perfect control over the actual taxes assessed. What the planner chooses is a \emph{target} tax $\overline{\tau}_i$ for each $i$, and the actual tax assessed is 
$$ 
  \tau_i = \overline{\tau}_i \eta_i,
$$ 
where $\eta_i$ is an independent random variable with mean $1$, variance $\nu^2$, and support bounded by $S \nu^2$ for some large $S$.\footnote{The assumption of bounded support is a technical one that can be relaxed at the cost of some additional details to be dealt with in proofs.} On average, the planner gets what she wants, but there are shocks that perturb the implemented tax some amount from its intended target. This $\tau_i$ represents \emph{implementation noise}, which for technical reasons enters in a multiplicative way.

The timing is that the planner sets the profile $(\overline{\tau}_i)_i$, the implementation noise $\eta_i$ is realized for every $i$, then equilibrium outcomes are determined, and payoffs are received.

We assume that the planner seeks to increase consumer surplus, but is risk-averse over this outcome. Let $\change{{\mathcal{C}}}$ denote  the change in consumer surplus, which is a random variable because the implemented tax is random.  The simple specification we use is that the planner's value function depends on the mean and variance of $\change{{\mathcal{C}}}$ as follows: \begin{equation} \mathcal{W} = \mathbb{E} \left[ {\change{\mathcal{C}}} \right] -  \frac{a}{2} \operatorname{Var} \left[ {\change{\mathcal{C}}} \right]. \label{eq:planner_W}\end{equation}  Thus, the planner cares about improving consumer surplus but is risk-averse and dislikes changes that lead to a highly variable consumer surplus. The parameter $a$ measures the extent of the planner's risk-aversion.  

As we will see, when $a$ is large, then the planner chooses a vector of taxes and subsidies whose norm is small. The risk-aversion in the planner's payoff function can be seen as capturing, for example, political considerations: if ex post the penalty for a harmful reform is greater than the reward for a successful one, then it is intuitive that a planner will be cautious. We will see that this means that the planner will pursue first-order gains that small tax schemes allow, while limiting their size to avoid exposure to risk. This feature also contributes to the tractability of the model, allowing us to capture the main economic insights in a closed-form solution of an optimization problem. 

\medskip

Denote by $ R$ the net revenue of the tax to the government. We require that the planner runs a balanced budge in expectation. Thus, the planner's problem, for a given level $a$ of risk-aversion, is
\begin{equation} \text{choose } \overline{\bm{\tau}} \text{ to } \operatorname{max} \mathcal{W} \text{ subject to }  \mathbb{E}[ R]=0.  \label{eq:plannersproblemnew} \tag{\text{T}(a)} \end{equation}

Our result on the form of the optimal intervention involves a statistic of the spillover matrix, which we now define. Let $\operatorname{var}[{\bm{\sigma}}]$ denote the sample variance of the eigenvalues $(\sigma_1,\ldots,\sigma_N)$, i.e., 
$$ 
\operatorname{var}[{\bm{\sigma}}] = \frac{1}{N}\sum_{i=1}^N \sigma_i^2 - \left(\frac{1}{N}\sum_{i=1}^N \sigma_i\right)^2.
$$
Since $\bm{D}$ is symmetric with $D_{ii}=-1$, it is easy to verify that
$$
  \operatorname{var}[{\bm{\sigma}}] = \frac{1}{N}\sum_{i \neq j}D_{ij}^2.
$$
Thus, $\operatorname{var}[{\bm{\sigma}}]$ captures the magnitude of the total strategic interaction---either substitutes or complements---of the different producers.

As we now show, the eigenvalue variance plays a key role in characterizing both the form of the optimal intervention and its efficacy.

\begin{thm}\label{prop:local}
For large enough $a$, the optimal policy taxes each eigenbundle $\ell$ with $\lambda_{\ell}<z$ and subsidizes each eigenbundle $\ell$ with $\lambda_{\ell}>z$,
where $z$ is the the shadow price of the budget-balance constraint and is equal to:
$$ z = \left[ 2+ \frac{1}{2} \operatorname{var}[{\bm{\sigma}}]\right]^{-1}.$$ 
Furthermore, in the limit $a \to \infty$ of small interventions, the optimal policy sets target taxes $\overline{\underline{\tau}}_\ell$ so that
\begin{itemize}
\item the tax raised from the $\ell^{\text{th}}$ eigenbundle satisfies
$$\lim_{a\to\infty} a \underline{q}_\ell \overline{\underline{\tau}}_\ell = {(z-\lambda_\ell)} (1+\sigma_\ell)^2,$$
\item the change in consumer surplus satisfies \begin{equation} \label{eq:change_surplus_small}\lim_{a\to\infty}a \change{\mathcal{C}}^* = N \frac{\operatorname{var}[\bm{\sigma}]}{4+\operatorname{var}[\bm{\sigma}]}.\end{equation} 
\end{itemize}
  \end{thm}

\medskip

The proof appears in Appendix \ref{proofs:small}. The following example illustrates the theorem. We then comment on the broad implications.

\begin{example}\label{Example:SI}

Consider a generalization of the Example 3.2 in which product 1 and product 2 are independent, but where we set $D_{31}=D_{32}=g$ with $g\in[-1/\sqrt{2},1/\sqrt{2}]$. When $g=0$, the three products are independent. When $g<0$, product 3 is a complement to both products 1 and 2, as we discussed earlier in Example \ref{Example:SI}; when $g>0$ product $3$ is a substitute to product $1$ and $2$. (For example, we may take product 1 to be a music player, product 2 to be a digital camera, and product 3 to be a smartphone.) The property that $\bm{D}$ is negative semi-definite means that $|g|\leq 1/\sqrt{2}$.

In this case, we have the following expressions for the eigenvalues of $\bm{D}$: $$\sigma_1=-1+\sqrt{2}|g| \qquad \sigma_2=-1 \qquad \sigma_3=-1-\sqrt{2}|g|.$$ The variance of the eigenvalues is $\operatorname{var}[\bm{\sigma}]=4g^2/3$. Furthermore, the shadow price of the budget balance constraint is $$z(g)=\frac{3}{6+2g^2},$$ and at the limit $a \to \infty$ of small intervention, the increase in consumer surplus is proportional to $$ \frac{g^2}{3+g^2}.$$

When all three products are independent, $g=0$, each eigenvalue is equal to $-1$, which means that the pass-through for each eigenbundle is the normalized monopoly price pass-through, $1/2$. In this case, budget neutral interventions are ineffective because any gains in consumer surplus due to subsidizing some eigenbundles can only be achieved by taxing some eigenbundles and that will create an equal loss in consumer surplus. Because budget-neutral interventions are ineffective, the shadow price of the budget neutral constraint is large, $z(0)=1/2$. Indeed, if we give a dollar to the planner, the planner will be able to pass through half of it to consumers.

As $|g|$ increases, the first eigenvalue becomes larger than $-1$, the second eigenvalue remains constant and the third becomes lower than $-1$. Now, the eigenbundles are heterogeneous with respect to their price pass-throughs: the price pass-through of the first eigenbundle becomes greater than $1/2$, $\lambda_1=1/(2+\sqrt{2}|g|)$ and the price pass-through of the last eigenbundle becomes lower than $1/2$, $\lambda_3=1/(2-\sqrt{2}|g|)$. The planner can now increase consumer surplus by taxing the third eigenbundle, which creates little reduction in consumer surplus, and use the tax revenue to subsidise the first eigenbundle, which creates large gains in consumer surplus. Indeed, the optimal intervention becomes more and more effective in increasing consumer surplus as $|g|$ increases and it reaches a (roughly) 15\% increase when $|g|=1/\sqrt{2}$. As the planner can do more and more by targeting taxes and subsidises, the shadow price of the budget neutral constraint is decreasing in $|g|$. 
\end{example}

We now draw out some general implications. First, the eigenvalue variance $\operatorname{var}[\bm{\sigma}]$ is a key statistic that captures the extent to which targeted tax/subsidy interventions are effective in an oligopoly setting. As this number grows, the consumer surplus gain achieved by a planner with the same level of risk aversion increases. At the same time, as the planner can increase the objective without running deficits, the shadow price of the budget neutral constraint $z$ decreases. 

The basic intuition is that an increase in the eigenvalue variance of $\bm{D}$ means that the price pass-through of the different eigenbundles becomes more heterogeneous. The planner can then tax eigenbundles with low price pass-through, thereby create a small decrease in consumer surplus, and then redistribute these resources to eigenbundles with high pass-throughs, which leads to a large increase in consumer surplus.

Second, recall that we can write the eigenvalue variance of $\bm{D}$ as $\operatorname{var}[\bm{\sigma}] =  \frac{1}{N}\sum_{i\neq j} D_{ij}^2.$ When products are independent then $-\bm{D}$ is the identity matrix and so the eigenvalue variance  is $0$ and so the Pigouvian leverage is zero. In contrast, complement and substitute relations across products, manifested in non-zero off-diagonal entries in the matrix $\bm{D}$, lead to strategic interaction in firms pricing behavior, and create scope for planner intervention. In this regard, the Pigouvian leverage does not depend on whether products are complements or substitutes, or a mix. It is only the intensity of demand interdependence of each pair of products that matters. A planner can achieve the same amount of improvement to consumer surplus no matter how the signs of various $D_{ij}$ are changed. 


\section{Global Pigouvian policy: the case of linear demand}\label{sec:linear}
The aim of this section is to show that the qualitative insights on optimal interventions are robust and extend beyond the setting of small tax reforms. To do so, we consider the case where interventions are not noisy, i.e., the variance of the implementation noise parameters $\eta_i$ is zero. In this case, the consumer surplus determined by a specific intervention is deterministic and the planner may very well choose large interventions.

We characterize the optimal intervention under the classical case of linear-quadratic specification of utility that gives rise to linear demands (Example \ref{ExQU}). The consumer has (gross) utility for consuming a bundle $\bm{q}$ of goods given by 
\begin{equation} 
\mathcal{V}(\bm{q}) = \bm{\beta}^\tr \bm{q} - \frac{1}{2}\bm{q}^\tr \bm{B} \bm{q}, \label{eq:payoff} 
\end{equation} 
where $\bm{\beta}\in \mathbb{R}^N_{+}$ and $\bm{B}$ is a given positive-definite matrix (\citet*{amir2017microeconomic}, \citet*{chone2020linear}, and \citet{vives1999oligopoly}). Under the assumption that the consumer has sufficient income, her optimal consumption choice induces linear demands $\bm{q}=\bm{B}^{-1}[\bm{\beta}-\bm{p}]$ with:  
\begin{equation} \partial q_i / \partial p_{j} = -(\bm{B}^{-1})_{ij} \quad \text{and} \quad \frac{\partial^2 q_i}{\partial p_{i} \partial p_{j}}=0 \quad \text{for any $i,j \in \mathcal{N}$}. \label{eq:final_demand} \end{equation}
In this case, the matrix $\bm{D}=-\bm{B}^{-1}$ and the normalization in Property A means that $D_{ii}=-1$ for all $i\in \mathcal{N}$. A standard calculation also implies that, under the consumer's optimal consumption, the consumer surplus is 
\begin{equation} 
\mathcal{C} = \frac{1}{2} \bm{q}^\tr \bm{B} \bm{q}. \label{eq:eqm_payoff} \end{equation}
Finally, the firms' pricing game under linear demand leads to the following equilibrium prices:
$$
[\bm{I}-\bm{D}]\bm{p}=-\bm{D}\bm{\beta}+\bm{c}.
$$

Let $\bm{\tau}=(\tau_1,\dots, \tau_n)$ be the profile of per-unit taxes on products introduced by the planner. 
We seek to solve the following problem:

\begin{align}\label{Program:Linear}
\begin{array} {ccc}
\underset{\bm{\tau}}{\text{Maximize}} &  \frac{1}{2} \bm{q}^\tr \bm{B} \bm{q} &\text{(Consumer surplus)}\\
\text{subject to} & \bm{q}=-\bm{D}(\bm{\beta}-\bm{p})& \text{(Optimal consumption)}\\
&[\bm{I}-\bm{D}]\bm{p}=-\bm{D}\bm{\beta}+\bm{c}+\bm{\tau}& \text{(Equilibrium prices)}\\
&\langle \bm{\tau},\bm{q}\rangle=0 & \text{(Budget balance)}
 \end{array}
\end{align}

Proposition \ref{prop:global} summarises the solution to this problem.

\begin{prop}\label{prop:global}
The optimal policy taxes each eigenbundle $\ell$ with $\lambda_{\ell}<z$, and subsidizes each eigenbundle $\ell$ with $\lambda_{\ell}>z$, where $z$ is the shadow price of the budget balance constraint and it is the unique positive solution of the following equation:
\begin{align}\label{eq:shadow}
\begin{array}{ccc}
\sum_{\ell}\left(\underline{q}_{\ell}^0\right)^2\frac{(z-\lambda_\ell)}{(1-\lambda_\ell)(2z-\lambda_\ell)^2}=0.
\end{array}
\end{align}
For generic $\bm{q}^0$, it satisfies $ \max\left\{\lambda_1,\lambda_n/2\right\}\leq z \leq \lambda_n. $ Finally, the optimal taxes/subsidizes are:
\begin{align}\label{eq:taxes}
\begin{array}{ccc}
\underline{\tau}_\ell=\frac{\underline{q}_{\ell}^0(z-\lambda_\ell)}{(1-\lambda_\ell)(2z-\lambda_\ell)} \quad \forall \ell\in \mathcal{N}.
\end{array}
\end{align}
\end{prop}

The proof appears in Appendix \ref{sec:proof_linear}.

Consistently with the optimal small intervention solution, in order to maximize consumer surplus, the planner taxes the eigenbundles with a low price pass-through (high eigenvalues) and uses the revenue generated to subsidize the eigenbundles with high price pass-through (low eigenvalues). The threshold that distinguishes eigenbundles which are taxed from the ones that are subsidised is the shadow price $z$. Generically, the shadow price must be larger than $\max\left(\lambda_1,\lambda_n/2\right)$ and it must be smaller than $\lambda_n$. In other words, the largest and smallest eigenvalues of the spillover matrix $\bm{D}$ determine an upper bound and a lower bound of the shadow price of public funds. This result is reminiscent of the role, in the small intervention case, of the sample variance of the eigenvalues in determining the value of $z$.

In relation to this, note also that, for the intervention to create an improvement in consumer surplus, products must have some degree of complementarity or substitutability. If all products are independent, i.e., $\bm{D}=-\bm{I}$, then eigenbundles have equal price pass-through and, therefore, the optimal policy is to maintain the status quo, i.e., $\underline{\tau}_\ell=0$ for all $\ell\in \mathcal{N}$.

We now illustrate the optimal policy within the example with three products, returning to the setting of Section \ref{sec:ex}.

\subsection{Example} \label{sec:ex-linear}
We assume that product 1 and product 2 are approximately independent, i.e., $D_{12}=D_{21}\approx 0$.\footnote{For $\bm{D}$ to be negative definite we need that $D_{12}=D_{21}\neq 0$. We report the results for the limit case in which $D_{12}=D_{21}\rightarrow 0$.} We also assume that product $3$ is equally complementary to product $1$ and $2$, $D_{13}=D_{12}=g>0$. In this case, the eigenbundles are independent of $g$ and are the same as the ones reported in Example \ref{sec:ex}. Furthermore, the eigenvalues are as in Example \ref{Example:SI}:  $$\sigma_1=1+\sqrt{2}|g| \qquad \sigma_2=1 \qquad \sigma_3=1-\sqrt{2}|g|.$$

We set the marginal costs for the three products to be zero, the demand parameters in \autoref{eq:payoff} to be $\bm{\beta}=(10, 10, 15)$, and the degree of complementarity to be $g=0.5$. In this case, \begin{itemize} 
\item the initial demand for the three products is $(1.51,1.51,4.2)$,
\item the price pass-throughs of the three eigenbundles are $(0.37,0.5,0.78)$ and 
\item the initial quantities of the eigenbundles are $(4.48,0,1.44)$. 
\end{itemize}

The shadow price of the budget constraint is $z=0.59$ and, therefore, the optimal policy taxes eigenbundle $1$, is neutral with respect to eigenbundle $2$ (because the initial quantity is zero) and subsidizes eigenbundle $3$. In particular, the tax revenue from eigenbundle $1$ (which is the same as the subsidy expenditure for eigenbundle $3$) is $6.3$. The value of this intervention can be measured by comparing the obtained increase in consumer surplus relative to the status quo. We obtain that consumer surplus goes up from $9.5$ to $10.8$. 

Consider a technological change that increases the complementarity between product $3$ and products $1$ and $2$; for concreteness, suppose that $g$ increases from $0.5$ to $0.55$. The difference between the largest and lowest eigenvalues increases  
and therefore the difference between the price pass-through of the two extreme eigenbundles increases as well: the price pass-through of the first eigenvalue decreases from $0.37$ to $0.36$, whereas the price pass-through of the third eigenvalue increases from $0.78$ to $0.82$. This implies that now taxing the first eigenbundle induces less distortion and subsidizing the third eigenbundle creates greater benefits. Indeed, we obtain that the technological change increases the shadow price $z$ of public funds  from $0.59$ to $0.7$ (meaning that the intervention is more valuable).
 These qualitative insights carry over when product $3$ is a substitute to products $1$ and $2$ (e.g., $g<0$).\footnote{However, the fact that prices are now strategic complements implies that the first eigenbundle (i.e., in which costs shocks are dampened the most) is now what before was the third eigenbundle $(0.5,0.5,-1/\sqrt{2})$, and the third one (i.e., in which costs shocks are dampened the least) is now what before was the first eigenbundle $(0.5,0.5,1/\sqrt{2})$.}

Translated back to the original product space, the logic of the optimal policy in this example is rather intricate, as we now discuss. When we are in the complements case, as in Section \ref{sec:ex}, the optimal policy strongly taxes products $1$ and $2$ (the software applications) and more lightly subsidizes product $3$ (the computer).  The purpose of such optimal policies is to leverage the strategic interactions between producers to induce a reduction in the price of the high-volume product $3$.  Taxing the software applications \emph{indirectly} incentivizes  the computer producer---via the resulting price increase in the prices of complements---to reduce its price.  (Recall prices are strategic substitutes in this case.)  
This reinforces the direct effect of the subsidy to computers. The computer subsidy must be kept fairly low, as otherwise (via the strategic substitutability in prices) we would end up \emph{lowering} the price of the applications, which would be counterproductive because of the aforementioned strategic substitutes. 

Consider now the case where the products are substitutes, (e.g., $g=-0.5$): the music player (product 1), the camera (product 2), and the smartphone (product 3). Smartphones remain the highest-volume product under our parameterization. Now, the optimal policy taxes smartphones and subsidizes music players and cameras. This is counterintuitive, because one might think that the planner would want to subsidize the product that consumers buy the most of, to improve consumer surplus.   The logic behind the optimal policy can again be understood via thinking about the strategic spillovers.  The subsidies to the specialized products \emph{indirectly} incentivize the smartphone producer---via the resulting price decrease in the specialized substitute products---to reduce its price. (The prices are strategic complements in this case.) This more than compensates for the effect of the tax directly levied on smartphones, which is used to raise revenue. Indeed, this scheme turns out to be better than subsidizing smartphones directly: if the planner were to implement such subsidies, they would have to be offered on the whole (large) quantity of smartphones sold, and end up being quite expensive per unit of welfare generated compared to the optimal intervention.

In the principal component decomposition, these intricate policies both correspond to the natural policy of subsidizing the high pass-through eigenbundles and taxing the low pass-through ones.

\section{Conclusion}

The paper studies firms interacting strategically in an environment with rich heterogeneity. We bring a network perspective to the interactions between the different firms. The key contribution is to decompose any cost change into components that have a very convenient form: for these components, there is an unambiguous ranking of pass-throughs from cost changes to equilibrium prices. We show how a policymaker can use this ranking to design a tax policy that maximizes consumer surplus: The optimal policy leverages the strategic interactions among producers (in particular, that taxes on some dimensions have a higher pass-through to equilibrium prices) to tax some producers in order to use the tax revenue to subsidize other producers. The effectiveness of such interventions is greater in markets with strong demand interdependence across products, as summarised by the eigenvalue variance of the matrix of spillovers $\bm{D}$.


\bibliographystyle{ACM-Reference-Format}
\bibliography{tax,svd}

\clearpage

\begin{appendix}
\begin{center}
\textbf{\LARGE{}Appendix}{\LARGE\par}
\par\end{center}

\section{Normalization of spillover matrix} \label{sec:normalization}

For any function $f:\mathbb{R}^n \to \mathbb{R}^m$, let $\mathcal{D} f(\bm{x})$ be the Jacobian matrix of the function evaluated at $\bm{x} \in \mathbb{R}^n$, whose $(i,j)$ entry is $\partial f_i/\partial x_j $, where $f_i$ denotes coordinate $i$ of the function.

Here we will be explicit about distinguishing quantity variables $\bm{q}$ from the corresponding demand function; to this end, we will write the function as $\mathfrak{q}$.

Consider the change of coordinates for quantities given by $\widetilde{q}_i = r_i q_i$. Keeping units of money fixed, \textbf{the corresponding prices are $\widetilde{p}_i = p_i/r_i$}. Let $\bm{R}$ be the diagonal matrix whose $(i,i)$ entry is $r_i$. With these new units, we can define a function $$\widetilde{\mathfrak{q}}(\widetilde{\bm{p}})= \bm{R} \mathfrak{q}(\bm{R} \widetilde{\bm{p}})$$ and by the chain rule we have that $$ \mathcal{D} \widetilde{\mathfrak{q}}(\widetilde{\bm{p}}) = \bm{R} \left[\mathcal{D} \mathfrak{q}(\bm{R} \widetilde{\bm{p}})\right] \bm{R}. $$

For a given demand function $\mathfrak{q} : \mathbb{R}^N \to \mathbb{R}^N$, recall $\bm{D}$ is defined to be $\mathcal{D} \mathfrak{q}(\bm{p}^*)$, where $\bm{p}^*$ are equilibrium prices, uniquely determined under our maintained assumptions. Here we will be explicit about the demand function and write $\bm{D}^{\mathfrak{q}}$ for $\mathcal{D} \mathfrak{q}(\bm{p}^*)$, It follows from this and the above paragraph that 
$$ \bm{D}^{\widetilde{\mathfrak{q}}} = \bm{R} \bm{D}^{\mathfrak{q}} \bm{R}. $$

Now set $r_i=1/\sqrt{\left|D^{\mathfrak{q}}_{ii}\right|}$. It is clear from the above formula that $\bm{D}^{\widetilde{\mathfrak{q}}}$ has $-1$ on the diagonal.

Thus, under a suitable choice of units, the matrix $\bm{D}$ may be assumed to have diagonal $-1$.

\section{Proofs for the small-intervention analysis} \label{proofs:small}

Denote by $\change{x}$ the change in any endogenous variable $x$ relative to the initial equilibrium for a given realization of random taxes.

We will state a useful lemma:
\begin{lem}\label{lem:EPT}
The pass-through from the marginal tax profile ${\bm{\tau}}$ to to consumer surplus, in the principal component basis, is
\begin{equation}
\change{\mathcal{C}}=-\underset{\ell}{\sum} \lambda_{\ell}\underline{q}_{\ell} \underline{\tau}_{\ell} + \mathcal{O}(\Vert\underline{\bm{\tau}}\Vert^2).\\ \label{U}
\end{equation}

\end{lem}
\begin{proof}[Proof of \autoref{lem:EPT}]

	From \autoref{eq:utilitymax}, and using the envelope theorem, the change in the utility of the representative consumer is $$\change{\mathcal{C}}=-\langle \change{\bm{p}},\bm{q} \rangle + \mathcal{O}(\Vert \change{\bm{p}} \Vert^2) .$$ Since $\bm{U}$ is orthonormal, we have that $\change{\mathcal{C}}=-\langle  \underline{\change{\bm{p}}},\underline{\bm{q}} \rangle + \mathcal{O}(\Vert\change{\underline{\bm{p}}}\Vert^2)$. The conclusion follows as in \autoref{prop:eig}. We use the bounded-support assumption on $\eta$ to ensure that as $\Vert \change{\bm{p}} \Vert^2$ becomes small, the comparative statics of the prices are characterized exactly by \autoref{prop:eig}.
\end{proof}

Recall that the planner's problem is 
\begin{equation} \text{choose } \overline{\bm{\tau}} \text{ to } \operatorname{max} \mathcal{W} \text{ subject to }  \mathbb{E}[ R]=0 \tag{$\text{\underline{T}}(a)$}  \label{eq:plannersproblemnew2}  \end{equation}

Consider an arbitrary small tax profile $\bm{\tau}$ and let $\widetilde{\bm{q}}$ be the new quantity and $\widetilde{\bm{p}}$ be the new price. Now, by equation (\ref{FOC}), under the normalization assumed by Property A, we have $$ \widetilde{\bm{q}} =  \widetilde{\bm{p}} -\bm{c} - \bm{\tau}.$$  Then the budget-balance condition can be transformed as follows:
\begin{align*}
     0 &= \mathbb{E}\left[ \sum_i \tau_i \widetilde{q}_i \right]  \\
    &= \mathbb{E}   \left[ \sum_i \tau_i \left( q_i + \dot{p}_i-\tau_i \right)\right]  \\
    &= \mathbb{E} \left[ \sum_i \left(\tau_i  q_i + \tau_i \dot{p}_i-\tau_i^2 \right)\right] 
\end{align*}
Transforming to the principal component basis, we get a clean formula in terms of the principal components for the budget-balance constraint:
\begin{equation} \mathbb{E} \left[ \sum_\ell \left( \underline{\tau}_\ell \underline{q}_\ell + \lambda_\ell \underline{\tau}_\ell^2 - \underline{\tau}_\ell^2 \right) \right]=0 \label{BB-nice} \end{equation}

Now, in order to write the maximand of the planner's problem in a nice way, we will recall two things. First, \autoref{lem:EPT} permits us to write $$ \change{\mathcal{C}} = \underbrace{- \sum_\ell  \lambda_\ell \underline{q}_\ell \underline{\tau}_\ell}_{\change{\mathcal{C}}_{\text{FO}}} + \mathcal{O}(\Vert\underline{\tau}\Vert^2).$$  Here we have defined $\change{\mathcal{C}}_{\text{FO}}$ to be the first-order term in the change. Second, recall our assumption $\tau_i = \overline{\tau}_i \eta_i$, where $\eta_i$ is has a mean of $1$. Assume without loss of generality\footnote{Since we are taking $a$ to be large, the limit does not depend on $\nu^2$.} that the variance of implemented taxes is $\nu^2=1$. We may write $$ \Var[\change{\mathcal{C}}_{\text{FO}}]= \sum_\ell  \lambda_\ell^2 \underline{q}_\ell^2 \underline{\overline{\tau}_\ell}^2.$$

We will solve an auxiliary problem of maximizing the planner's utility as if the change in consumer surplus were only the first-order change. That is, we will maximize  $\mathcal{W}(\change{\mathcal{C}}_{\text{FO}})$ subject to the budget constraint in expectation. At the end, we will show the approximation of replacing $\change{\mathcal{C}}$ by $\change{\mathcal{C}}_{\text{FO}}$ does not change the characterization of the optimal intervention as $a \to \infty$.

Letting $x$ be the Lagrange multiplier on the constraint that $\mathbb{E}[R]=0$,  the Lagrangian of the planner's problem
(\ref{eq:plannersproblemnew2}) can be written as\footnote{Recall the planner's utility function from \autoref{eq:planner_W}.} 
\begin{align*} \mathcal{L} &= \mathcal{W}(\change{\mathcal{C}}_{\text{FO}}) + z \mathbb{E}[R]  \\
&= -\sum_\ell \left[ \lambda_\ell \underline{q}_\ell \underline{\overline{\tau}}_\ell + \frac{a}{2}\lambda_\ell^2 \underline{q}_\ell^2  \overline{\underline{\tau}}_\ell^2 - z \underline{q}_\ell \overline{\underline{\tau}}_\ell \right]
\end{align*}

Differentiating with respect to each $\overline{\underline{\tau}}_\ell$ gives \begin{equation} \label{eq:char_interventions}\underline{q}_\ell \overline{\underline{\tau}}_\ell = \frac{z-\lambda_\ell}{a \lambda_\ell^2}. \end{equation}

\subsection{Solution for the Lagrange multiplier \texorpdfstring{\textit{z}}{$z$}}
 Plugging these into the formula (\ref{BB-nice}) above yields the following sequence of equivalent expressions
\begin{align*}
    \sum_\ell \overline{\tau}_\ell q_\ell &= 0 \\
    \sum_\ell \left( \frac{z-\lambda_\ell}{\lambda_\ell^2}\right) &= 0 \\
    z \sum_\ell \lambda_\ell^{-2} - \sum_\ell \lambda_\ell^{-1} &= 0.
\end{align*}

Recalling the definition $ \lambda_\ell:=(1-\sigma_\ell)^{-1}$,  we can deduce that \begin{equation} z = \frac{\frac{1}{N}\sum_\ell (1 - \sigma_\ell)}{\frac{1}{N}\sum_\ell (1-\sigma_\ell)^2} \label{eq:z_version1} \end{equation} This gives an explicit form for the Lagrange multiplier. 

To understand this formula better, define $\mu_\ell = \lambda_\ell^{-1}=1-\sigma_\ell$. Let $\mu$ be a random variable equal to $\mu_\ell$, with $\ell$ selected uniformly at random. Then we can rewrite
$$ z = \frac{\mathbb{E}[\mu]}{\mathbb{E}[\mu^2] } = \frac{\mathbb{E}[\mu]}{\mathbb{E}[\mu]^2 + \text{Var}[\mu] }. $$

Simplifying this requires a preliminary lemma.

\begin{lem} \label{lem:mu} $\sum_\ell \mu_\ell = 2N$.  \end{lem}
	\begin{proof} Recall that $\sigma_\ell$ is the $\ell$th eigenvalue of $\bm{D}$. By Property A, the trace of $\bm{D}$ is $-N$. This implies that the sum of its eigenvalues, $\sum_\ell \sigma_\ell$, is equal to $-N$. Thus the sum of $\mu_\ell=1-\sigma_\ell$ is equal to $2N$. \end{proof}

We may then characterize the Lagrange multiplier as 
\begin{equation} \label{eq:simple_z} z = \left[2 + \frac{1}{2} \operatorname{var}[\bm{\sigma}] \right]^{-1}. \end{equation}

\subsection{A formula for the change in social surplus} We can also plug in our solution for the optimal intervention to obtain an asymptotically correct formula for $\change{\mathcal{C}}$.

First, recall from Lemma \ref{lem:EPT} that
$$ \change{\mathcal{C}}_{\text{FO}}=-\underset{\ell}{\sum} \lambda_{\ell}\underline{q}_{\ell} \underline{\tau}_{\ell}. $$ Plugging in (\ref{eq:char_interventions}), 
we get 
$$ E(\change{\mathcal{C}}_{\text{FO}})= -\sum_\ell \frac{z-\lambda_\ell}{a \lambda_\ell}.$$ Thus \begin{align*} a E(\change{\mathcal{C}}_{\text{FO}})&= -\sum_\ell \left( \frac{z}{\lambda_\ell}-1\right) \\ &= N-\sum_\ell  \frac{\mu_\ell}{2+\frac{1}{2}\Var[\sigma]} \\ &= N-N \frac{2}{2+\frac{1}{2}\Var[\sigma]} && \text{using \autoref{lem:mu}}. \end{align*}
From this we conclude that 
$$ \frac{a}{N} E(\change{\mathcal{C}}_{\text{FO}}) = 1 - \frac{1}{1+\frac{1}{4}\Var[\sigma]}. $$

\subsection{Dispensing with the approximate consumer surplus}
It remains to show that neglecting higher-order terms in the optimization problem---solving the problem with $\change{\mathcal{C}}_{\text{FO}}$ replacing $\change{\mathcal{C}}$---does not change the asymptotic characterization of the optimal intervention.

	We will prove the result by studying an equivalent optimization problem using Berge's Theorem of the Maximum. Let $\widecheck{{\underline{\bm{\tau}}}} = a {\overline{\underline{\bm{\tau}}}}$.  We will now define a rescaled version of the problem, $\widecheck{\underline{\text{T}}}(a)$.
	\begin{align*}
\text{choose } \widecheck{\underline{\bm{\tau}}} \text{ to } & \operatorname{max} a{\mathcal{W}}(\change{\underline{\mathcal{C}}}(a^{-1} \widecheck{{\underline{\bm{\tau}}}}))   \tag{$\widecheck{\text{\underline{T}}}(a)$}\\
	& \text{ subject to }  \mathbb{E}[ R]=0 
	\end{align*} Here $\underline{\mathcal{C}}$ is the change consumer surplus when the tax profile is represented in the diagonal basis.
	This is clearly equivalent to the original problem. Let $\widecheck{{\underline{\bm{\tau}}}}^*(a)$ be the (possibly set-valued) solution for risk-aversion level $a>0$.
	
	The  problem $\widecheck{\underline{\text{T}}}(a)$ is not yet defined at $a=\infty$, but we now define it there. Let the objective at $a=\infty$ be the limit of $a\mathcal{W}(a^{-1} \widecheck{{\underline{\bm{\tau}}}})$  as $a \to  \infty$; it is easy to check that this limit is actually $$F(\widecheck{{\bm{\tau}}})=-\sum_\ell \left[ \lambda_\ell \underline{q}_\ell \underline{\widecheck{{\tau}}}_\ell + \frac{1}{2}\lambda_\ell^2 \underline{q}_\ell^2  {\widecheck{\underline{\tau}}}_\ell^2  \right].$$ Let the constraint be \begin{equation} \sum_\ell  \widecheck{\underline{\tau}}_\ell \underline{q}_\ell  =0, \label{BB-verynice} \end{equation} which is (\ref{BB-nice}) with the second-order terms dropped.
	
When we restrict the optimization problem $\widecheck{\underline{\text{T}}}(a)$ to any compact set $\mathcal{K}$ containing the origin,  the of Berge's Theorem of the Maximum are satisfied: The constraint correspondence is continuous at $a=\infty$, while the objective function is jointly continuous in  $a$ and $\widecheck{{\underline{\bm{\tau}}}}$. The Theorem of the Maximum therefore implies that the maximized value is continuous at $a = \infty$. Because the convergence of the objective is actually uniform on $\mathcal{K}$, this is possible if and only if $\widecheck{{\underline{\bm{\tau}}}}$ approaches a solution of the problem
	\begin{align*}
	\max_{\widecheck{{\underline{\bm{\tau}}}}}   \text{ } & F(\widecheck{{\bm{\tau}}}) \\
	& \text{s.t. } \sum_\ell  \widecheck{\underline{\tau}}_\ell \underline{q}_\ell  =0 . \nonumber
	\end{align*}
Taking $\mathcal{K}$ to contain the interior solutions of this problem, this shows that the conclusions claimed hold without the first-order approximation.

\subsection{Another way to write the variance of \textsigma}
Here is another way to write the variance of the eigenvalues, for what it's worth:
Note that the average of the $\sigma_\ell$ is $-1$, so that 
\begin{align*} 
\text{var}(\bm{\sigma}) &= \frac{1}{N}\sum_\ell (\sigma_\ell+1)^2 \\
 &= \frac{1}{N} \text{trace}(\bm{I}+\bm{\Sigma})^2 \\
 &= \frac{1}{N} \text{trace}\left[(\bm{I}+\bm{D})^2\right] && \text{trace invariant under similarity} \\
 &= \frac{1}{N} \text{trace}(\bm{I}+2\bm{D}+\bm{D}^2)  \\
 &= \frac{1}{N}\left(N - 2N + \text{trace}(\bm{D}^2)  \right) && \text{trace of $\bm{D}$ is $N$} \\
\\
 &=-1+ \frac{1}{N} \sum_{i, j} D_{ij}^2 
= \frac{1}{N} \sum_{i\neq j} D_{ij}^2 \end{align*}

\section{Analysis of the linear model} \label{sec:proof_linear}
The representative consumer chooses a bundle $\bm{q}$ to maximize
\[
\mathcal{C}=\bm{\beta}^\tr \bm{q}-\frac{1}{2}\bm{q}^\tr \bm{B}\bm{q}-\dd{\bm{p}}{\bm{q}}.
 \]
Recall that we are assuming that the representative consumer has enough income to buy the optimal consumption bundle. The first-order conditions for optimal consumption imply that, at an interior solution,
\begin{equation}\label{eq:L_q}
\bm{q}=\bm{B}^{-1}\left(\bm{\beta}-\bm{p}\right)
\end{equation}
The consumer's utility under optimal consumption (i.e., the consumer surplus) is therefore
\[
\mathcal{C}=\frac{1}{2}\bm{q}^\tr\bm{B}\bm{q}.
\]
The matrix of spillovers $\bm{D}$ defined in the main text is $\bm{D}=-\bm{B}^{-1}$ and the normalization in Property A sets $\bm{D}_{ii}=-(\bm{B}^{-1})_{ii}=-1$. Furthermore, given linear demand, the equilibrium in prices (under the assumption of uniqueness and interiority) is
\[
\bm{p}=\bm{q}+\bm{c}
\]
Using equilibrium demand \ref{eq:L_q} we have
\[
\bm{p}=-\bm{D}\left(\bm{\beta}-\bm{p}\right)+\bm{c}
\]
or, after defining $\bm{a}:=-\bm{D}\bm{\beta}$, %
\[
[\bm{I}-\bm{D}]\bm{p}=\bm{a}+\bm{c}
\]
The intervention problem specified in \ref{Program:Linear} is analogous to:
\begin{align}
\begin{array} {ccc}
\underset{\bm{\tau}}{\text{Maximize}} &  \frac{1}{2}\bm{q}^\tr\bm{B}\bm{q} \\
\text{subject to} & \bm{q}=\bm{a}+\bm{D}\bm{p},\\
&[\bm{I}-\bm{D}]\bm{p}=\bm{a}+\bm{c}+\bm{\tau}\\
&\dd{\bm{\tau}}{\bm{q}}\geq 0
\end{array}
\end{align}

\textbf{Proof of Proposition \ref{prop:global}.} Using the decomposition $\bm{D}=\bm{U}\bm{\Sigma}\bm{U}^\tr$, we can rewrite the problem as:
\begin{align}
\begin{array} {ccc}
\underset{\bm{\underline{\tau}}}{\text{Maximize}} &  -\frac{1}{2}\bm{\underline{q}}^\tr\bm{\Sigma^{-1}}\bm{\underline{q}} \\
\text{subject to} & \bm{\underline{q}}=\bm{\underline{a}}+\bm{\Sigma}\bm{\underline{p}},\\
&[\bm{I}-\bm{\Sigma}]\bm{\underline{p}}=\bm{\underline{a}}+\bm{\underline{c}}+\bm{\underline{\tau}}\\
&\dd{\bm{\underline{\tau}}}{\bm{\underline{q}}}\geq 0
\end{array}
\end{align}
Combining the demand equation with the equilibrium price equation (i.e., the first two constrains), and defining the quantity consumed in equilibrium pre-intervention by $\bm{\underline{q}}^{0}=\bm{\underline{a}}+\bm{\Sigma}\left[\bm{I}-\bm{\Sigma}\right]^{-1}(\bm{\underline{a}}+\bm{\underline{c}})$, we obtain that the quantities consumed after intervention $\bm{\tau}$ are
\[
\bm{\underline{q}}=\bm{\underline{q}}^{0}+\bm{\Sigma}\left[\bm{I}-\bm{\Sigma}\right]^{-1}\bm{\underline{\tau}}.
\]
Therefore the taxation problem is
\begin{align}
\begin{array} {ccc}
\underset{\bm{\underline{\tau}}}{\text{Maximize}} &  \frac{1}{2}\bm{\underline{q}}^\tr\bm{\Sigma^{-1}}\bm{\underline{q}} \\
\text{subject to} &\bm{\underline{q}}=\bm{\underline{q}^0}-\bm{\Sigma}\left[\bm{I}+\bm{\Sigma}\right]^{-1}\bm{\underline{\tau}}\\
&\dd{\bm{\underline{\tau}}}{\bm{\underline{q}}}\geq 0
\end{array}
\end{align}
which can be rewritten as follows:
\begin{align} \label{eq:problem_for_lagrangian}
\begin{array} {ccc}
\underset{\bm{\underline{\tau}}}{\text{Maximize}} &  -\frac{1}{2}\sum_{\ell}\frac{1}{\sigma_\ell}\left(\underline{q}^0_\ell+\frac{\sigma_\ell}{1-\sigma_{\ell}}\underline{\tau}_\ell\right)^2\\
\text{subject to} & -\sum_{\ell} \underline{\tau}_\ell \left(\underline{q}^0_\ell+\frac{\sigma_\ell}{1-\sigma_{\ell}}\underline{\tau}_\ell\right)\leq 0
\end{array}
\end{align}
Letting $z \geq 0$ be the multiplier on the constraint (the shadow price of public funds), we have the Lagrangian\footnote{We first write $-z$ as the coefficient on the derivative of the constraint, and then cancel the minus sign on the constraint in (\ref{eq:problem_for_lagrangian}) with the minus sign on $z$.}
\[
\mathcal{L}=-\frac{1}{2}\sum_{\ell}\frac{1}{\sigma_\ell}\left(\underline{q}^0_\ell+\frac{\sigma_\ell}{1-\sigma_{\ell}}\underline{\tau}_\ell\right)^2+z  \sum_{\ell} \underline{\tau}_\ell \left(\underline{q}^0_\ell+\frac{\sigma_\ell}{1-\sigma_{\ell}}\underline{\tau}_\ell\right).
\]
 We can choose the signs of the eigenbundles so that $\underline{q}^0_\ell \geq 0$ for every $\ell\in \mathcal{N}$. Let us also order the eigenvalues so that $\sigma_1 \geq \sigma_2 \geq \cdots \geq \sigma_n$. Recall that $\lambda_\ell=1/(1-\sigma_\ell)$.

\begin{lem} \label{lem:tau} The necessary conditions for an optimum, $\frac{\partial{\mathcal{L}}}{{\partial \underline{\tau}_\ell}}=0$ for all $\ell\in \mathcal{N}$, imply the optimal taxes are characterized by $$\underline{\tau}_\ell=\frac{\underline{q}_\ell^0}{1-\lambda_\ell} \frac{z-\lambda_\ell}{2z-\lambda_\ell}.$$ \end{lem}

\begin{proof}
Writing out $\frac{\partial{\mathcal{L}}}{{\partial \underline{\tau}_\ell}}=0$ we have
\[
-\left(\underline{q}^0_\ell+\frac{\sigma_\ell}{1-\sigma_{\ell}}\underline{\tau}_\ell\right)\frac{1}{1-\sigma_\ell}+z\left(\underline{q}^0_\ell+2\frac{\sigma_\ell}{1-\sigma_\ell}\tau_\ell\right)=0
\]
or
\[
\underline{\tau}_\ell\frac{-\sigma_\ell}{1-\sigma_\ell}\left(\frac{1}{1-\sigma_\ell}-2z\right)=\underline{q}^0_\ell\left(\frac{1}{1-\sigma_\ell}-z\right)
\]
or
\[
\underline{\tau}_\ell\frac{-\sigma_\ell}{1-\sigma_\ell}=\underline{q}^0_\ell\frac{z(1-\sigma_\ell)-1}{2z(1-\sigma_\ell)-1}.
\]
Solving for $\underline{\tau}_\ell$ and rewriting in terms of $\lambda_\ell$ yields the claim. \end{proof}

\begin{lem} \label{lem:qfinal} The quantity of each eigenbundle at the planner's optimum is $$ \underline{q}_\ell = \underline{q}^0_\ell \frac{z}{2z-\lambda_\ell}.$$ \end{lem} 

\begin{proof} Using the formula $$\bm{\underline{q}}=\bm{\underline{q}}^{0}+\bm{\Sigma}\left[\bm{I}-\bm{\Sigma}\right]^{-1}\bm{\underline{\tau}}$$ in the principal component basis, along with the expression for optimal taxes from Lemma \ref{lem:tau},
\begin{align*} \underline{q}_\ell&=\underline{q}^0_\ell+\frac{\sigma_\ell}{1-\sigma_{\ell}}\underline{q}_\ell^0\frac{1}{1+\lambda_\ell}\frac{z-\lambda_\ell}{2z-\lambda_\ell}
\\ &=  \underline{q}^0_\ell \left(1+\frac{\sigma_\ell}{1-\sigma_{\ell}}\frac{1}{1-\lambda_\ell}\frac{z-\lambda_\ell}{2z-\lambda_\ell}\right) \\ &= \underline{q}_\ell^0  \left(1-\frac{z-\lambda_\ell}{2z-\lambda_\ell}\right)  \\ &= \underline{q}_\ell^0 \cdot \frac{z}{2z-\lambda_\ell}.  \end{align*}
\end{proof}

\begin{lem} \label{lem:taxfrombundle}
The tax revenue raised from the $\ell^{\text{th}}$ eigenbundle is:
$$ \underline{q}_\ell\underline{\tau}_\ell=  (\underline{q}_\ell^0)^2 \frac{1}{1-\lambda_\ell}\frac{z(z-\lambda_\ell)}{(2z-\lambda_\ell)^2}.$$
\end{lem}
\begin{proof}
This is immediate using the expression for optimal taxes from \autoref{lem:tau} and the expression for quantities after the intervention from \autoref{lem:qfinal}.
\end{proof}

Using \Autoref{lem:taxfrombundle}, we can rewrite the binding budget constraint  $\sum_\ell \underline{q}_\ell\underline{\tau}_\ell=0$ as \begin{equation}\sum_\ell (\underline{q}_\ell^0)^2 \frac{1}{1-\lambda_\ell}\frac{z(z-\lambda_\ell)}{(2z-\lambda_\ell)^2}=0. \label{budget-rewritten}\end{equation}  As long as $\underline{\bm{q}}^0 \neq 0$,  it follows that either $z$ must be equal to some $\lambda_\ell$ (with  $\underline{q}_{\ell'}^0=0$ for all $\ell' \neq \ell$), or that $z-\lambda_\ell$ is positive for some $\ell$ and negative for some other $\ell$. In either case, it follows that $z \in [\lambda_1,\lambda_n].$ 

Each term of the sum induces an asymptote at $z=\lambda_\ell/2$. Let $\ell'$ be the smallest index so that $\underline{q}_{\ell'}^0 \neq 0$. (Recall $\lambda_n \leq \lambda_{n-1} \leq \ldots \leq \lambda_1$.) The left-hand side of the above equation asymptotes to $-\infty$ as $z$ decreases to $\lambda_{\ell'}/2$, and equals to $\sum_\ell (\underline{q}_\ell^0)^2\frac{1}{4(1-\lambda_\ell)}>0$ as $z$ increases to $+\infty$. Thus there is a solution $z\geq \lambda_{\ell'}/2$ of the rewritten budget constraint, \autoref{budget-rewritten}. Moreover, any solution $z >0$ of the equation along with the optimal tax formulas above gives (by the KKT theorem) a solution to our optimization problem. But the optimization problem is strictly convex, so that its solution must be unique, and thus there is a unique value of $z$ that solves \autoref{budget-rewritten}. 

If $\underline{q}_\ell \neq 0$ for all $\ell$, which holds generically, we can conclude that $$ \max\left(\lambda_n,\frac{\lambda_1}{2}\right)\leq z \leq \lambda_1 .$$ Indeed, all the inequalities discussed above must be strict.

Finally, since $2z-\lambda_\ell>0$ for each $\ell$ (by the above argument) we have from Lemma 2 that $\underline{q}_\ell$, the post-tax quantity, must have the same sign as $\underline{q}_\ell^0$, the pre-tax quantity.

This concludes the proof of Proposition \ref{prop:global}.

\end{appendix}

\end{document}